\documentclass[a4paper,10pt]{article}
\usepackage{amsmath}
\usepackage{amssymb}
\usepackage{amsthm}
\usepackage{accents}

\newcommand{\wt}{\tilde}
\newcommand{\wh}{\hat}
\renewcommand{\th}[1]{\wh{\wt{#1}}}
\newcommand{\wb}{\bar}
\newcommand{\ut}[1]{{\underaccent{\tilde}{#1}}}
\newcommand{\uh}[1]{{\underaccent{\hat}{#1}}}
\newcommand{\flr}[1]{{\left \lfloor{#1}\right \rfloor}}
\newcommand{\oT}{\mathrm{T}}
\newcommand{\oS}{\mathrm{S}}

\newcommand{\mm}{\mathop{\sf{m}}\nolimits}

\newcommand{\ub}[1]{{\underaccent{\bar}{#1}}}

\newcommand{\bt}[1]{\wt{\wb{#1}}}
\newcommand{\hb}[1]{\wh{\wb{#1}}}

\newtheorem{thm}{Theorem}[section]
\newtheorem{prop}[thm]{Proposition}

\newcommand{\sn}{\mathop{\mathrm{sn}}\nolimits}
\newcommand{\cn}{\mathop{\mathrm{cn}}\nolimits}
\newcommand{\dn}{\mathop{\mathrm{dn}}\nolimits}
\newcommand{\ETA}{\mathop{\mathrm{H}}\nolimits}
\newcommand{\THE}{\mathop{\mathrm{\Theta}}\nolimits}

\newcommand{\crossratio}[4]{
\frac{(#1-#2)(#3-#4)}{(#1-#3)(#2-#4)}
}

\begin{document}

\begin{center}
{\LARGE{The tau function for ABS equations}}
\vskip10mm
{\large James Atkinson\footnote{Contact: james.l.atkinson@gmail.com}, October 28, 2024}\\
\vskip3mm
\end{center}
\noindent
{\bf Abstract.}
The tau-function for quad-equations from the ABS classification is briefly explained.
It is an auxiliary variable that systematically linearises the B\"acklund chain.
Many equations have the same tau function and are unified by transformations constructed on that basis.
Simple examples are included.

\section{Constructing tau}
Recall the ABS equations, listed in Table \ref{list}, are considered in general on quads of $\mathbb{Z}^d$ with dynamics defined by a polynomial of degree one in four variables,
\begin{equation}
Q_{\alpha,\beta}(u,\wt{u},\wh{u},\th{u})=0.\label{qeq}
\end{equation}
Here $(\wt{\phantom{u}},\alpha)$ and $(\wh{\phantom{u}},\beta)$ are representative lattice translation/parameter pairs taken from a list, say $(\oS_1,\alpha_1),\ldots,(\oS_d,\alpha_d)$.

An associated function $\tau$ \cite{AN,AJ} is constructed from a solution $u$ of (\ref{qeq})
by solving auxiliary equations,
\begin{equation}
\frac{\th{\tau}\ut{\tau}}{\wt{\tau}\ut{\wh{\tau}}} = -\frac
{Q_{\alpha,\beta}(u,\wt{u},\wh{u},\wh{\ut{u}})}
{Q_{\alpha,\beta}(u,\ut{u},\wh{u},{\th{u}})},
\qquad
\frac{\th{\tau}\uh{\tau}}{\wh{\tau}\uh{\wt{\tau}}} = -\frac
{Q_{\alpha,\beta}(u,\wt{u},\wh{u},\wt{\uh{u}})}
{Q_{\alpha,\beta}(u,\wt{u},\uh{u},{\th{u}})},\label{taudef}
\end{equation}
where $u\mapsto \ut{u}$ indicates the shift opposite to $u\mapsto\wt{u}$ so that $\wt{\ut{u}}=u$. 
Only a particular solution of (\ref{taudef}) is required, the general solution is obtained by using symmetries of two kinds:
\begin{equation}\label{transform1}
\tau(n,m,\ldots)\rightarrow F(n)\tau(n,m,\ldots), \quad n,m,\ldots\in\mathbb{Z},
\end{equation}
where $F$ is an arbitrary function in one lattice direction,
and
\begin{equation}\label{transform2}
\tau(n,m,\ldots)\rightarrow c^{nm-2{\left \lfloor{nm/2}\right \rfloor}}\tau(n,m,\ldots), \quad n,m,\ldots\in\mathbb{Z},
\end{equation}
where $c$ is an arbitrary constant associated with one pair of lattice directions.

\begin{table}
\begin{tabular}{l|l}
\hline
& $Q_{\alpha,\beta}(u,\wt{u},\wh{u},\th{u})$ \\
\hline
Q4&
$\sn(\alpha)(u\wt{u}+\wh{u}\wh{\wt{u}})-\sn(\beta)(u\wh{u}+\wt{u}\wh{\wt{u}})-$\\&$\qquad \sn(\alpha-\beta)[u\wh{\wt{u}}+\wt{u}\wh{u}-k\sn(\alpha)\sn(\beta)(1+u\wt{u}\wh{u}\wh{\wt{u}})]$\\
Q3& 
$(\alpha-1/\alpha)(u\wt{u}+\wh{u}\th{u})-(\beta-1/\beta)(u\wh{u}+\wt{u}\th{u})-$\\&$ \qquad(\alpha/\beta-\beta/\alpha)[\wt{u}\wh{u}+u\th{u}+\delta^2(\alpha-1/\alpha)(\beta-1/\beta)/4]$\\
Q2& 
$\alpha(u-\wh{u})(\wt{u}-\th{u})-\beta(u-\wt{u})(\wh{u}-\th{u})+$\\&$ \qquad\alpha\beta(\alpha-\beta)(u+\wt{u}+\wh{u}+\th{u}-\alpha^2+\alpha\beta-\beta^2)$\\
Q1& $\alpha(u-\wh{u})(\wt{u}-\th{u})-\beta(u-\wt{u})(\wh{u}-\th{u})+\delta^2 \alpha\beta(\alpha-\beta)$\\
A2&$(\alpha-1/\alpha)(u\wh{u}+\wt{u}\th{u})-(\beta-1/\beta)(u\wt{u}+\wh{u}\th{u}) - (\alpha/\beta-\beta/\alpha)(1+u\wt{u}\wh{u}\th{u})$ \\
A1&$\alpha(u+\wh{u})(\wt{u}+\th{u})-\beta(u+\wt{u})(\wh{u}+\th{u})- \delta^2\alpha\beta(\alpha-\beta)$\\
H3& $\alpha(u\wt{u}+\wh{u}\th{u})-\beta(u\wh{u}+\wt{u}\th{u})+\delta(\alpha^2-\beta^2)$\\
H2& $(u-\th{u})(\wt{u}-\wh{u})-(\alpha-\beta)(u+\wt{u}+\wh{u}+\th{u}+\alpha+\beta)$\\
H1& $(u-\th{u})(\wt{u}-\wh{u})+\alpha-\beta$\\
\hline
\end{tabular}
\caption{ABS equations \cite{ABS,ABS2}, $k\in\mathbb{C}$ is the Jacobi modulus, $\delta\in\{0,1\}$.}
\label{list}
\end{table}

\subsection{Example}
Consider H1,
and the solution on $\mathbb{Z}^{2}$,
\begin{equation}\label{H1sol}
u= np+mq, \quad p^2=\alpha, \quad q^2=\beta.
\end{equation}
The equations for $\tau$ (\ref{taudef}) are,
\begin{equation}
\frac{\th{\tau}\ut{\tau}}{\wt{\tau}\ut{\wh{\tau}}} = \frac{q-p}{q+p},
\qquad
\frac{\th{\tau}\uh{\tau}}{\wh{\tau}\uh{\wt{\tau}}} = \frac{p-q}{p+q},
\end{equation}
which imply that
\begin{equation}\label{H1tau}
\tau = (-1)^{\flr{n/2}m}\left(\frac{p-q}{p+q}\right)^{\flr{nm/2}},
\end{equation}
up to admissible transformations.

Extended to $\mathbb{Z}^d$, $u$ and $\tau$ are
\begin{equation}\label{utauex}
u = \sum_i n_ip_i, \quad \tau = \prod_{i<j}(-1)^{\flr{n_i/2}n_j}\left(\frac{p_i-p_j}{p_i+p_j}\right)^{\flr{n_in_j/2}},
\end{equation}
where $p_i^2=\alpha_i$ and $i,j\in\{1,\ldots,d\}$.

\section{Integrating the B\"acklund chain}
Successful construction of $\tau$ for a multidimensional solution $u$, integrates the 
B\"acklund chain seeded by $u$ for any equation from the list in Table \ref{list}.
The details follow.

Distinguish lattice and spectral dimensions initially by notation:
suppose on $\mathbb{Z}^{d+N}$ that $u$ and $\tau$ satisfy (\ref{qeq}) and (\ref{taudef}) with translation and parameter pairs denoted by,
\begin{equation}\label{nsplit}
(\oS_1,\alpha_1),\ldots,(\oS_d,\alpha_d),(\oT_1,\lambda_1),\ldots,(\oT_N,\lambda_N).
\end{equation}
Introduce the restriction to $\mathbb{Z}^d$ of a generic function $F$ defined on $\mathbb{Z}^{d+N}$, as
\begin{equation}
\left.F\right\vert_{\mathbb{Z}^d}:=F(n_1,\ldots,n_d,0,\ldots,0).
\end{equation}
The main result concerns functions defined on $\mathbb{Z}^d$,
\begin{equation}\label{ab2}
u_I := \left.\frac{\left(\prod_{i\in I}\left[a_i\oT_i+b_i\oT_i^{-1}\right]\right)u\tau}{\left(\prod_{i\in I}\left[a_i\oT_i+b_i\oT_i^{-1}\right]\right)\tau}\right\vert_{\mathbb{Z}^d}, \quad I\subset\{1,\ldots,N\},
\end{equation}
where $(a_1:b_1),\ldots,(a_N:b_N)\in\mathbb{P}^1$ are freely chosen
and correspond to the integration constants for the B\"acklund chain.
\begin{prop}\label{solprop}
Functions $u_I$ satisfy the quad-equation on $\mathbb{Z}^d$,
\begin{equation}
Q_{\alpha,\beta}(u_I,\wt{u}_I,\wh{u}_I,\th{u}_I)=0,
\end{equation}
they satisfy the B\"acklund relations $u_{I}\xrightarrow{\lambda_i} u_{I\oplus \{i\}}$ in each direction in $\mathbb{Z}^d$,
\begin{equation}\label{bteq}
Q_{\alpha,\lambda_i}(u_I,\wt{u}_I,u_{I\oplus\{i\}},\wt{u}_{I\oplus\{i\}})=0, \quad i\in\{1,\ldots,N\},
\end{equation}
and are related algebraically by superposition,
\begin{equation}\label{speq}
Q_{\lambda_i,\lambda_j}(u_I,u_{I\oplus \{i\}},u_{I\oplus \{j\}},u_{I\oplus\{i,j\}})=0, \quad i,j\in\{1,\ldots,N\}, i\neq j,
\end{equation}
where $\oplus$ indicates symmetric difference, $I\oplus J=(I\cup J) \setminus (I\cap J)$.
\end{prop}
\begin{proof}[Outline of proof]
First, instance $I=\{\}$ of (\ref{bteq}) is confirmed by substitution, recovering the equations for $\tau$,
and in turn, instance $I=\{\}$ of (\ref{speq}) is a consequence that can be confirmed in the same way.
For calculations this is sufficient, because the remaining equations follow by induction on $|I|$. 
Specifically, if $u$ and $\tau$ are known on $\mathbb{Z}^{d+1}$, then the linear contraction of the solution,
\begin{equation}
u_{\{1\}}\tau_{\{1\}} = [a_1\oT_1+b_1\oT^{-1}]u\tau|_{\mathbb{Z}^d}, \quad \tau_{\{1\}} = [a_1\oT_1+b_1\oT_1^{-1}]\tau|_{\mathbb{Z}^d}
\end{equation}
constructs $u_{\{1\}}$ and $\tau_{\{1\}}$ on $\mathbb{Z}^d$, forming the induction principle.
The details, including the handling of certain singularities in the solutions, are in the original papers \cite{AN,AJ}.
\end{proof}

If $u'=(au+b)/(cu+d)$, then $u'$ satisfies an equation in the same class,
so it is meaningful to ask for its tau function, and it is easily checked that $\tau'=(cu+d)\tau$. 
However,
\begin{equation}\label{ab3}
u_I' = \frac{au_I+b}{cu_I+d}= \left.\frac{\left(\prod_{i\in I}\left[a_i\oT_i+b_i\oT_i^{-1}\right]\right)(au+b)\tau}{\left(\prod_{i\in I}\left[a_i\oT_i+b_i\oT_i^{-1}\right]\right)(cu+d)\tau}\right\vert_{\mathbb{Z}^d}.
\end{equation}
I.e., although $\tau$ itself is not invariant, the multisoliton formula is an invariant function of the pair $(u\tau,\tau)$.

\section{Equations for tau}
Any system characterising $\tau$ or $u\tau$ in multidimensions 
must have the linear contraction property mentioned. 
Specifically, given free constants $a$ and $b$, the transform,
\begin{equation}
\tau\rightarrow (a\wt{\tau}+b\ut{\tau})|_{\mathbb{Z}^{d}}
\end{equation}
 produces a new solution on $\mathbb{Z}^{d}$ from a known solution on $\mathbb{Z}^{d+1}$, i.e., excluding equations involving shifts in the distinguished lattice direction.
It is of interest to exhibit examples with this direct integrability feature.
The equations characterising $\tau$ separate the ABS equations.
\begin{prop}\label{poorprop}
The tau function (\ref{taudef}) for each equation in Table \ref{list}, with the exception of Q4 and A2, is characterised on $\mathbb{Z}^2$ by the Hirota KdV equation \cite{hirota-0},
\begin{equation}\label{taurel}
\ut{\uh{\tau}}\wt{\tau}\wh{\tau} + 
\wh{\wt{\tau}}\ut{\tau}\uh{\tau} =
\ut{\wh{\tau}}\wt{\tau}\uh{\tau} +
\wt{\uh{\tau}}\ut{\tau}\wh{\tau}.
\end{equation}
\end{prop}
\begin{proof}
Substitution confirms that (\ref{taurel}) is a consequence of (\ref{qeq}) and (\ref{taudef}).
To show that $\tau$ is otherwise free, consider the following transformation.
Let
\begin{equation}
f_{\alpha,\beta}(w,x,y):=\partial_z Q_{\alpha,\beta}(w,x,y,z),
\end{equation} 
and write
\begin{align}
\frac{\th{\tau}\ut{\tau}}{\wt{\tau}\ut{\wh{\tau}}} = \frac
{f_{\alpha,\beta}(u,\wt{u},\wh{u})}
{f_{\alpha,\beta}(u,\ut{u},\wh{u})},
&&
\frac{\th{\tau}\uh{\tau}}{\wh{\tau}\uh{\wt{\tau}}} = \frac
{f_{\alpha,\beta}(u,\wt{u},\wh{u})}
{f_{\alpha,\beta}(u,\wt{u},\uh{u})},
\\
Q_{\alpha,\beta}(u,\ut{u},\wh{u},\ut{\wh{u}})=0,
&&
Q_{\alpha,\beta}(u,\wt{u},\uh{u},\uh{\wt{u}})=0.
\end{align}
If $\tau$ on $\mathbb{Z}^2$ is fixed, then the birational maps respectively defined,
\begin{equation}\label{map2}
(\wt{u},u,\wh{u}) \mapsto (u,\ut{u},\ut{\wh{u}}), \quad
(\wt{u},u,\wh{u}) \mapsto (\uh{\wt{u}},\uh{u},u),
\end{equation}
have (\ref{taurel}) as their compatibility condition, and the emerging function $u$ satisfies equations (\ref{qeq}) and (\ref{taudef}).
\end{proof}
Equation (\ref{taurel}) can be written as
\begin{equation}\label{hkdv}
v+\th{v} = 1/\wt{v} + 1/\wh{v}, \quad v=\frac{\tau\th{\tau}}{\wt{\tau}\wh{\tau}},
\end{equation}
which is equivalent to the form in \cite{hirota-0}.
Bilinear relations are expected \cite{kajota} but only after introducing a non-autonomous scaling factor (gauge factor). 
The most direct statement of the more refined higher dimensional characterisation uses an intermediary Schwarzian variable \cite{James2}.
\begin{prop}\label{sv}
The tau function (\ref{taudef}) of each equation in Table \ref{list}, with the exception of Q4 and A2, is characterised on $\mathbb{Z}^d$ for $d>2$ by the system,
\begin{equation}
\crossratio{\varphi}{\wt{\varphi}}{\wh{\varphi}}{\th{\varphi}}=\frac{t(\beta)-t(\gamma)}{t(\alpha)-t(\gamma)},\qquad 
\varphi=\frac{\wb{\tau}}{\ub{\tau}},
\label{constCR}
\end{equation}
in terms of the lattice Schwarzian KdV equation \cite{NC}, or Q1$_{\delta=0}$, for $\varphi$, where 
\begin{equation}\label{plist}
t(\alpha) = \left\{
\begin{array}{ll}
1/(1-\alpha^2), &(Q3),\\
1/\alpha, &(Q2,Q1,A1),\\
1/\alpha^2, &(H3),\\
\alpha,  &(H2,H1),
\end{array}
\right.
\end{equation}
and $(\wb{\phantom{u}},\gamma)$ denote the shift and parameter in a generic third lattice direction.
\end{prop}\begin{proof}
Substitution confirms that (\ref{constCR}) is a consequence of (\ref{qeq}) and (\ref{taudef}).
To show that $\tau$ on $\mathbb{Z}^d$ is otherwise free, consider a transformation as  follows.
Write
\begin{align}
\frac{\wt{\varphi}}{\varphi} = \frac
{f_{\alpha,\gamma}(u,\wt{u},\wb{u})}
{f_{\alpha,\gamma}(u,\wt{u},\ub{u})},
&&
\frac{\wh{\varphi}}{\varphi} = \frac
{f_{\beta,\gamma}(u,\wh{u},\wb{u})}
{f_{\beta,\gamma}(u,\wh{u},\ub{u})},\label{phidef2}
\\
Q_{\alpha,\gamma}(u,\wt{u},\wb{u},\bt{u})=0, &&Q_{\beta,\gamma}(u,\wh{u},\wb{u},\hb{u})=0,\label{aux1}\\
Q_{\alpha,\gamma}(u,\wt{u},\ub{u},\wt{\ub{u}})=0, &&Q_{\beta,\gamma}(u,\wh{u},\ub{u},\wh{\ub{u}})=0.\label{aux2}
\end{align}
For fixed $\varphi$ on $\mathbb{Z}^{d-1}$ this system determines birational maps,
\begin{equation}\label{map}
(\ub{u},u,\wb{u})\mapsto (\wt{\ub{u}},\wt{u},\bt{u}), \quad (\ub{u},u,\wb{u})\mapsto (\wh{\ub{u}},\wh{u},\hb{u}),
\end{equation}
whose pairwise compatibility condition is (\ref{constCR}) for $\varphi$ on $\mathbb{Z}^{d-1}$, and the emerging function $u$ defined on $\mathbb{Z}^{d-1}\times (-1,0,1)$ satisfies (\ref{qeq}).
Because $d$ was arbitrary the claim follows by restriction to $\mathbb{Z}^{d-1}$.
\end{proof}

For the equations mentioned, the composition $u\mapsto \tau \mapsto u'$ produces a solution $u'$, which is a solution of the same equation, or a different one, and an example will be included at the end of this section.
However, by construction, $\tau'=\tau$ without loss of generality, so nothing more will be gained if the procedure is repeated.
I.e., the quad-equations mentioned in Proposition \ref{sv} have the same tau function.
On the other hand, the statement that two equations have a different tau function does not mean they should be distinguished:
the example A2 is equivalent to Q3$_{\delta=0}$, but as mentioned already, only affine point transformations on $u$ leave $\tau$ unaltered.

Like $\tau$, the quantity $\wb{\tau}/\ub{\tau}$ is also not invariant under fractional-linear point transformations of $u$.
An invariant Schwarzian variable  \cite{Weiss0,James2} is
\begin{equation}\label{sv2}
\varphi:=\frac{\wb{\tau}(\wb{u}-w)}{\ub{\tau}(\ub{u}-w)},
\end{equation}
where $w$ is any given singular solution of the quad equation.
$\varphi$ in (\ref{sv2}) is characterised by equations that are invariant under fractional linear transformations, and satisfies the lattice Schwarzian KP equation \cite{DN,bog}.
If $w$ is constant, then $\varphi$ satisfies the lattice Schwarzian KdV equation, and 
all equations in Table \ref{list} have a singular solution that is constant throughout $\mathbb{Z}^d$, with the exception of $A2$ and $Q4$. 
The canonical forms in Table \ref{list} are such that the constant singular solution is $w=\infty$ in all cases, so substituting this into (\ref{sv2})  offers an explanation for the isolation of $\tau$ leading to Proposition \ref{sv}.

\subsection{Example}
Given the particular tau function (\ref{H1tau}), the associated Schwarzian variable is
\begin{equation}\label{svex}
\varphi =\frac{\wb{\tau}}{\ub{\tau}}= \left(\frac{p-r}{p+r}\right)^n\left(\frac{q-r}{q+r}\right)^m,
\end{equation}
where $(\wb{\phantom{u}},\gamma)$ are the auxiliary shift and parameter, and $r^2=\gamma$ extending (\ref{H1sol}).
Substitution shows that
\begin{equation}\label{crpw}
\crossratio{\varphi}{\wt{\varphi}}{\wh{\varphi}}{\th{\varphi}}=\frac{q^2-r^2}{p^2-r^2},
\end{equation}
confirming (\ref{constCR}).
A solution $u$ whose corresponding Schwarzian variable is (\ref{svex}) was obtained for each of the quad-equations mentioned in Proposition \ref{sv} in \cite{NAH}, and are recalled here as the simplest example for Proposition \ref{sv}.

For H1 and H2 let
\begin{align}
v=np+mq, \qquad \alpha=p^2, \quad \beta=q^2.\label{H1param}
\end{align}
H1 solution:
\begin{equation}\label{H1sym}
u=(A+1/A)(v+B)/2 + (-1)^{n+m}(A-1/A)(v+C)/2.
\end{equation}
H2 solution:
\begin{equation}\label{H2sol}
u=(v+B)^2+(-1)^{n+m}A(2v+C) - A^2.
\end{equation}

For H3 let
\begin{align}
\theta=[\alpha(p-1)]^n[\beta(q-1)]^m, \qquad 1/\alpha^2=1-p^2, \quad 1/\beta^2=1-q^2.\label{H3param}
\end{align}
H3 solution:
\begin{equation}
u = A\theta + B/\theta + (-1)^{n+m}(C\theta + D/\theta), \quad 4(AB-CD)-\delta=0.
\end{equation}

For Q1 and Q2 let
\begin{align}
&w=np\alpha + mq\beta,\qquad
\rho = \left(\frac{p-1}{p+1}\right)^n\left(\frac{q-1}{q+1}\right)^m,\\
&1/\alpha = p^2-1, \quad 1/\beta = q^2-1.\label{Q1param}
\end{align}
Q1 solution:
\begin{equation}
u = A(w+D) + (1+B\rho)(1+C/\rho)/4 , \quad A^2+BC-\delta^2=0.
\end{equation}
Q2 solution:
\begin{equation}
u = (w+A)^2 + (1+B\rho)(1+C/\rho)/4.
\end{equation}

For Q3 let $a,b$ be constants, although only $a/b$ is essential, and 
\begin{align}
&\digamma = \left(\alpha\frac{p-a}{p+b}\right)^n\left(\beta\frac{q-a}{q+b}\right)^m, 
\qquad \rho = \left(\frac{p-a}{p+a}\right)^n\left(\frac{q-a}{q+a}\right)^m,\\
& 1/(1-\alpha^2) = (p^2-a^2)/(b^2-a^2), \quad 1/(1-\beta^2) = (q^2-a^2)/(b^2-a^2).\label{Q3param}
\end{align}
Q3 solution:
\begin{equation}
u=A\digamma+B/\digamma + C\digamma/\rho + D\rho/\digamma, \quad \frac{AB}{(a-b)^2}-\frac{CD}{(a+b)^2}+\frac{\delta^2}{16ab} = 0.
\end{equation}

Where they appear, $A,B,C,D$ are constants, which are either free or subject to the single constraint indicated.
All of these solutions have the same tau function (\ref{H1tau}), and a sufficient number of constants to satisfy initial data for the system of maps (\ref{map}) and (\ref{map2}).
\section{Reduced equations for tau}
The equations for tau will be reduced in order, i.e., reduced to equations on a smaller stencil. This comes at the cost of introducing an auxiliary gauge factor.

Define the polynomial $f_{\alpha,\beta}(w,x,y)=\partial_zQ_{\alpha,\beta}(w,x,y,z)$ as before.
Equations (\ref{taudef}) for $\tau$ are then recast using the assumption that $u$ satisfies the quad-equation,
\begin{equation}
\begin{split}
\frac{\th{\tau}\ut{\tau}}{\wt{\tau}\ut{\wh{\tau}}} = 
\frac
{f_{\alpha,\beta}(u,\wt{u},\wh{u})}
{f_{\alpha,\beta}(u,\ut{u},\wh{u})},
\qquad
\frac{\th{\tau}\uh{\tau}}{\wh{\tau}\uh{\wt{\tau}}} = 
\frac
{f_{\alpha,\beta}(u,\wt{u},\wh{u})}
{f_{\alpha,\beta}(u,\wt{u},\uh{u})},
\\
\frac{\wt{\uh{\tau}}\ut{\tau}}{\wt{\tau}\ut{\uh{\tau}}} = 
\frac
{f_{\alpha,\beta}(u,\wt{u},\uh{u})}
{f_{\alpha,\beta}(u,\ut{u},\uh{u})},
\qquad
\frac{\ut{\wh{\tau}}\uh{\tau}}{\wh{\tau}\ut{\uh{\tau}}} = 
\frac
{f_{\alpha,\beta}(u,\ut{u},\wh{u})}
{f_{\alpha,\beta}(u,\ut{u},\uh{u})}.\label{taudef3}
\end{split}
\end{equation}
The doubling of equations corresponds to possible lattice sites on which values of $u$ are eliminated.
An alternative writing of the same equations is
\begin{equation}\label{splitrel}
\frac{f_{\alpha,\beta}(u,\wt{u},\wh{u})}{v} 
= \ut{v} f_{\alpha,\beta}(u,\ut{u},\wh{u})
= \uh{v} f_{\alpha,\beta}(u,\wt{u},\uh{u})
= \frac{f_{\alpha,\beta}(u,\ut{u},\uh{u})}{\ut{\uh{v}}} = w,
\end{equation}
in terms of the Hirota variable,
\begin{equation}
v = \frac{\tau\th{\tau}}{\wt{\tau}\wh{\tau}}.
\end{equation}
The new function $w$ introduced in (\ref{splitrel}) will facilitate a partial integration of the system. 
Eliminating $v$ the equations (\ref{splitrel}) imply,
\begin{equation}\label{vdef}
v= 
\frac{f_{\alpha,\beta}(u,\wt{u},\wh{u})}{w} = 
\frac{\wt{w}}{f_{\alpha,\beta}(\wt{u},u,\th{u})} =
\frac{\wh{w}}{f_{\alpha,\beta}(\wh{u},\th{u},u)} = 
\frac{f_{\alpha,\beta}(\th{u},\wh{u},\wt{u})}{\th{w}},
\end{equation}
establishing dynamical equations for $w$ on quad edges,
\begin{equation}\label{vsys}
w\wt{w} = d(\alpha,\beta)h_{\alpha}(u,\wt{u}), \quad w\wh{w} = d(\beta,\alpha)h_\beta(u,\wh{u}).
\end{equation}
Simplification from (\ref{vdef}) to (\ref{vsys}) uses the assumption that $u$ satisfies the quad equation, together with the polynomial identity \cite{ABS},
\begin{multline}
[\partial_y Q_{\alpha,\beta}(w,x,y,z)][\partial_z Q_{\alpha,\beta}(w,x,y,z)]\\-Q_{\alpha,\beta}(w,x,y,z)[\partial_y\partial_z Q_{\alpha,\beta}(w,x,y,z)]
= d(\alpha,\beta)h_\alpha(w,x),\label{gdiscrim}
\end{multline}
for a constant $d(\alpha,\beta)=-d(\beta,\alpha)$ and a biquadratic polynomial $h_\alpha(w,x)$, see Table \ref{list2} for the list of constants.

\begin{table}
\begin{center}
\begin{tabular}{l|l}
\hline
&  $d(\alpha,\beta)$ \\
\hline
Q4& $2\sn(\alpha)\sn(\beta)\sn(\alpha-\beta)$\\
Q3& $(1-1/\alpha^2)(1-1/\beta^2)(\alpha^2-\beta^2)$\\
Q2& $4\alpha\beta(\alpha-\beta)$ \\
Q1& $2\alpha\beta(\alpha-\beta)$\\ 
A2& $(1-1/\beta^2)(1-1/\alpha^2)(\beta^2-\alpha^2)$\\
A1& $2\beta\alpha(\beta-\alpha)$ \\
H3& $\beta^2-\alpha^2$ \\
H2& $\beta-\alpha$ \\
H1& $\beta-\alpha$ \\
\hline
\end{tabular}
\caption{Skew-symmetric constants}
\label{list2}
\end{center}
\end{table}

System (\ref{vsys}) and the dimensional skew-symmetry of (\ref{splitrel}) indicates the function $w$ has a decomposition in potentials $U$ and $\mu$,
\begin{equation}
w=U\frac{\ut{\mu}\wh{\mu}}{\mu\wh{\ut{\mu}}}.
\end{equation}
First, $U$ satisfies the multidimensional system \cite{Adl,ABS},
\begin{equation}\label{Ueq}
U\wt{U} = h_\alpha(u,\wt{u}), \qquad U\wh{U} = h_\beta(u,\wh{u}).
\end{equation}
There is an interpretation $U=\partial_xu$ in terms of dependence of $u$ on an independent variable $x$, except for H1 and H3$_{\delta=0}$.
Second, $\mu$ is a gauge factor defined by,
\begin{equation}
\frac{\th{\mu}\ut{\mu}}{\wt{\mu}\ut{\wh{\mu}}} = {d(\alpha,\beta)},
\qquad
\frac{\th{\mu}\uh{\mu}}{\wh{\mu}\uh{\wt{\mu}}} = {d(\beta,\alpha)}. \label{mudef}
\end{equation}
This is an autonomous multidimensional system with the same form as the original system (\ref{taudef}) for $\tau$, depending on constants in Table \ref{list2}.

In summary, the reduced equations for $\tau$ (\ref{splitrel}) have the following meaning.
\begin{prop}
Suppose $u$ satisfies a quad-equation in Table \ref{list}. If $\tau$ is the associated tau function and $\mu$ satisfies (\ref{mudef}), then
\begin{equation}\label{Uassign}
 U = \frac{\wt{\sigma}\wh{\sigma}}{\sigma\th{\sigma}}\frac{f_{\alpha,\beta}(u,\wt{u},\wh{u})}{d(\alpha,\beta)}, \quad \sigma = \tau/\mu,
\end{equation}
satisfies (\ref{Ueq}).
Conversely, if $U$ satisfies (\ref{Ueq}), then $\tau$ is determined by (\ref{Uassign}).
\end{prop}
The quad-equation for $\sigma$ can be an advantage, as it is an alternative to embedding a quad-graph into $\mathbb{Z}^d$ allowing an initial value problem to be considered directly \cite{AdVeQ,pip}.

The reduced equations for $\tau$ provide a connection between Proposition \ref{sv} and the Sato framework via the well known equations of Miwa \cite{Miwa}.
\begin{prop}\label{bilinear}
The tau function (\ref{taudef}) of each equation in Table \ref{list}, with the exception of Q4 and A2, is a function on $\mathbb{Z}^d$ for $d>2$ of the form $\tau=\sigma\nu$, where $\sigma$ on $\mathbb{Z}^d$ satisfies the following reduction of the Miwa equation,
\begin{equation}\label{miwa}
\left[
\begin{array}{cccc}
0 & t(\alpha)-t(\beta) & t(\beta)-t(\gamma) & t(\gamma)-t(\alpha) \\
t(\alpha)-t(\beta) & 0 & -1 & 1\\
t(\beta)-t(\gamma) & 1 & 0 & -1\\
t(\gamma)-t(\alpha) & -1 & 1 & 0
\end{array}
\right]
\left[
\begin{array}{l}
\sigma\th{\wb{\sigma}}\\ 
\wb{\sigma}\th{\sigma}\\
\wt{\sigma}\hb{\sigma}\\
\wh{\sigma}\bt{\sigma}
\end{array}
\right]
=
\left[
\begin{array}{l}
0\\0\\0\\0
\end{array}
\right],
\end{equation}
and 
$\nu$ on $\mathbb{Z}^d$ is a solution of the multiplicatively linear system 
\begin{equation}\label{nudef}
\frac{\th{\nu}\ut{\nu}}{\wt{\nu}\ut{\wh{\nu}}} = {t(\beta)-t(\alpha)},
\qquad
\frac{\th{\nu}\uh{\nu}}{\wh{\nu}\uh{\wt{\nu}}} = {t(\alpha)-t(\beta)},
\end{equation}
with $t$ as defined in (\ref{plist}).
\end{prop}
\begin{proof}
The corresponding condition on $\tau$ is,
\begin{equation}\label{sigeq}
\left[
\begin{array}{cccc}
0&\wt{a}&\wh{b}&\wb{c}\\
\wt{a} & 0 & c & \hb{b} \\
\wh{b} & \bt{c} & 0 & a \\
\wb{c} & b & \th{a} & 0
\end{array}
\right]
\left[
\begin{array}{l}
\tau\th{\wb{\tau}}\\ 
\wb{\tau}\th{\tau}\\
\wt{\tau}\hb{\tau}\\
\wh{\tau}\bt{\tau}
\end{array}
\right]
=
\left[
\begin{array}{l}
0\\0\\0\\0
\end{array}
\right],
\end{equation}
where the rank-two matrix has coefficients in terms of
\begin{equation}\label{abc}
a = \frac{\nu\th{\nu}}{\wt{\nu}{\wh{\nu}}},\quad
b = \frac{\nu\hb{\nu}}{\wh{\nu}{\wb{\nu}}},\quad
c = \frac{\nu\bt{\nu}}{\wb{\nu}{\wt{\nu}}}.
\end{equation}
The necessity of (\ref{sigeq}) can be checked by glancing at Table \ref{list}, where the H1 polynomial satisfies identities,
\begin{equation}
\begin{split}
f_{\alpha,\beta}(u,\wt{u},\wh{u}) + f_{\beta,\gamma}(u,\wh{u},\wb{u}) + f_{\gamma,\alpha}(u,\wb{u},\wt{u}) = 0,\\
f_{\alpha,\beta}(u,\wt{u},\wh{u}) + f_{\beta,\gamma}(u,\wh{u},\ub{u}) + f_{\gamma,\alpha}(u,\ub{u},\wt{u}) = 0,
\end{split}
\end{equation}
and so on.
Combined with (\ref{splitrel}) they convert directly to the bilinear identities for $\tau$, so, according
to Proposition \ref{sv}, must hold for the common tau function of all mentioned equations.
That relations (\ref{sigeq}) imply (\ref{constCR}) can be checked by calculation using the defining properties of $\nu$, in particular that $\th{a}=-a$, $\wt{a}=-\wh{a}$, $a\wt{a}=t(\beta)-t(\alpha)$, $\wb{a}=a$, and so on.
\end{proof}
In particular, given any tau function for the mentioned equations, there is a unique gauge factor, $\nu$ satisfying (\ref{nudef}), such that $\sigma = \tau/\nu$ satisfies (\ref{miwa}).
Furthermore, linear contraction on (\ref{sigeq}) acts trivially on $\nu$, so descends to an action on $\sigma$ as follows (constants a and b are absourbed into $\nu$).
\begin{prop} 
If $\sigma$ satisfies (\ref{miwa}) on $\mathbb{Z}^{d+1}$, then
\begin{equation}
\sigma' = \left.\frac{\wt{\nu}\wt{\sigma} + \ut{\nu}\ut{\sigma}}{\nu}\right|_{\mathbb{Z}^d}
\end{equation}
satisfies (\ref{miwa}) on $\mathbb{Z}^d$, where $\nu$ on $\mathbb{Z}^{d+1}$ is any solution of (\ref{nudef}).
\end{prop}
The linear contraction on $\tau$ has the advantage of being cleanly separated from the symmetries (\ref{transform1}) and (\ref{transform2}) that alter $\nu$, but in this form on $\sigma$ it can be recognised as a B\"acklund transformation for the Miwa equation that is compatible with the reduction.

\subsection{Discussion}
The relationship between the tau function and bilinear forms proposed for the ABS equations \cite{HZ,WZM} has not been investigated, but they are traditionally very close ideas.

Characterisation of $\tau$ for A2 and Q4 is qualitatively different, as in these cases there is no effective difference between the characterisation of $\tau$ and the characterisation of $u\tau$.
It is therefore reasonable to be interested in the isolation of $u\tau$ even for the simplest equations.

\subsection{Example}
A solution of the Miwa equation reduction (\ref{miwa}) is
\begin{equation}
\sigma = \prod_{i<j}(p_i+p_j)^{-n_in_j}, \qquad t(\alpha_i) = p_i^2, \quad i,j\in\{1,\ldots,d\}.
\end{equation}
This solution is related to the ongoing 
tau function example (\ref{utauex}),
\begin{equation}
\tau = \prod_{i<j}(-1)^{\flr{n_i/2}n_j}\left(\frac{p_i-p_j}{p_i+p_j}\right)^{\flr{n_in_j/2}},
\end{equation}
because $\nu=\tau/\sigma$ satisfies system (\ref{nudef}) with $t(\alpha_i) = p_i^2$.

\section{Extended example: A2, F$_I$ and Q4$^*$ solutions}
Solutions of A2 inherent to the equation rather than its connection to Q3$_{\delta=0}$ are considered.
Corresponding solutions of $F_I$ and $Q4^*$ are immediate due to known transformations,
which are motivation to consider the A2 equation in its own right.

Different polynomials from Table \ref{list} will be used simultaneously, so define
\begin{multline}
R_{p,q}(w,x,y,z):=\\
(p-1/p)(wy+xz) -(q-1/q)(wx+yz)-(p/q-q/p)(1+wxyz),\label{A2}
\end{multline}
which is the A2 polynomial with a relabelling $\alpha\rightarrow p$ and $\beta\rightarrow q$.
\subsection{Seed}\label{ss}
A seed solution can be found as the simplest periodic reduction \cite{weiss}, a solution mapped to itself by the B\"acklund transformation.
Specifically,
\begin{equation}
R_{p,k}(u,\wt{u},u,\wt{u}) = 0, \quad
R_{q,k}(u,\wh{u},u,\wh{u}) = 0,
\end{equation}
where $k$ is a free parameter.
The solution is
\begin{equation}\label{seed}
u=\sqrt{k}\sn(\xi)=\frac{\ETA(\xi)}{\THE(\xi)}, \quad \xi=\xi_0+n\alpha +m\beta,
\end{equation}
where
\begin{equation}\label{eparams}
p=k\frac{\cn(\alpha)}{\dn(\alpha)}, \quad q=k\frac{\cn(\beta)}{\dn(\beta)}.
\end{equation}
The Jacobi modulus $k$ appearing in the solution is the B\"acklund parameter in the stationary direction.
Equality
\begin{equation}\label{rsol}
R_{p,q}(u,\wt{u},\wh{u},\th{u})=0,
\end{equation}
can be checked directly, but the following procedure was developed in \cite{AN} and also works for degenerations from Q4.

The Q4 polynomial in Jacobi form \cite{Hie} is taken directly from Table \ref{list},
\begin{multline}
Q_{\alpha,\beta}(w,x,y,z)=\sn(\alpha)(wx+yz)-\sn(\beta)(wy+xz)\\
-\sn(\alpha-\beta)[xy+wz-k\sn(\alpha)\sn(\beta)(1+wxyz)].
\label{Q4}
\end{multline}
An identity can be confirmed,
\begin{multline}\label{ident}
R_{p,q}(w,x,y,z) = \frac{1-k^2}{2\sqrt{k}}\frac{1-k^2\sn^2(\alpha)\sn^2(\beta)}{\sn(\alpha)\sn'(\alpha)\sn(\beta)\sn'(\beta)}\times \\\left[\sn(\beta+\alpha)Q_{\alpha,\beta}(w,x,y,z)+\sn(\beta-\alpha)Q_{\alpha,-\beta}(w,x,y,z)\right],
\end{multline}
where $p,q$ are related to $\alpha,\beta$ by (\ref{eparams}).
Because (\ref{seed}) is the singular solution of Q4 \cite{ABS2}, both terms in (\ref{ident}) vanish on it, confirming (\ref{rsol}).

\subsection{Tau function}
Identity (\ref{ident}) is also the first step to integrating the equations (\ref{taudef}) for $\tau$ in this example.
It yields 
\begin{equation}\label{taudef2}
\frac{\th{\tau}\ut{\tau}}{\wt{\tau}\ut{\wh{\tau}}} 
=-\frac
{R_{p,q}(u,\wt{u},\wh{u},\wh{\ut{u}})}
{R_{p,q}(u,\ut{u},\wh{u},{\th{u}})}\\
=-
\frac
{Q_{\alpha,-\beta}(u,\wt{u},\wh{u},\wh{\ut{u}})}
{Q_{\alpha,\beta}(u,\ut{u},\wh{u},{\th{u}})}
\frac{\sn(\beta-\alpha)}{\sn(\beta+\alpha)}.
\end{equation}
Here the properties of (\ref{seed}) as the Q4 singular solution are used: 
\begin{equation}\label{p0}
Q_{\alpha,\beta}(u,\wt{u},\wh{u},z)=0, \quad Q_{\alpha,-\beta}(u,z,\wh{u},\th{u})=0,
\end{equation}
hold identically as polynomials in $z$.

Indeed, generalising this,
\begin{equation}
Q_{\alpha,\beta}(u,\wt{u},y,z)=\sn(\alpha)t(\xi,\beta,\alpha-\beta)[\wh{u}-y][\wt{\uh{u}}-z],\label{id1}
\end{equation}
identically as a polynomial in $y$ and $z$, where 
\begin{multline}
t(a,b,c):= 1+k^2\sn(a)\sn(b)\sn(c)\sn(a+b+c)\\ = \frac{\THE(0)\THE(a+b)\THE(b+c)\THE(c+a)}{\THE(a)\THE(b)\THE(c)\THE(a+b+c)}.\label{tdef}
\end{multline}
The second identity in (\ref{p0}) may also be obtained from (\ref{id1}), by using the change $\beta\rightarrow -\beta$ and the quad symmetry of the polynomial.
Similarly, (\ref{id1}) can also be cast into each of the following two forms:
\begin{equation}
\begin{split}
Q_{\alpha,-\beta}(u,\wt{u},y,z)&=\sn(\alpha)t(\xi,-\beta,\alpha+\beta)[\uh{u}-y][\th{u}-z],\\
Q_{\alpha,\beta}(u,\ut{u},y,z)&=\sn(\alpha)t(\ut{\xi},\beta,\alpha-\beta)[\uh{u}-y][\ut{\wh{u}}-z].
\end{split}
\end{equation}
Direct application of these formulae reduces (\ref{taudef2}) to
\begin{equation}\label{final-tau}
\frac{\th{\tau}\ut{\tau}}{\wt{\tau}\ut{\wh{\tau}}} =
\frac{\THE(\th{\xi})\THE(\ut{\xi})\ETA(\beta-\alpha)}
{\THE(\wt{\xi})\THE(\wh{\ut{\xi}})\ETA(\beta+\alpha)}.
\end{equation}
The system (\ref{final-tau}), coupled with its counterpart obtained by interchanging $(\wt{\phantom{u}},\alpha)$ and $(\wh{\phantom{u}},\beta)$, integrates to
\begin{equation}\label{mdt}
\tau = (-1)^{\flr{n/2}m}\THE(\xi)\left(\frac{\ETA(\alpha-\beta)}{\ETA(\alpha+\beta)}\right)^\flr{nm/2},
\end{equation}
up to the standard transformations.

\subsection{Multidimensional extension}
The multidimensional seed and tau function are given by obvious extensions
\begin{equation}\label{mmdt}
u = \frac{\ETA(\xi)}{\THE(\xi)}, \quad \tau = \THE(\xi)\prod_{1\le i< j \le d}(-1)^{\flr{n_i/2}n_j}\left(\frac{\ETA(\alpha_i-\alpha_j)}{\ETA(\alpha_i+\alpha_j)}\right)^\flr{n_in_j/2},
\end{equation}
where now $\xi = \xi_0+n_1\alpha_1+\cdots +n_d\alpha_d$.

Splitting the lattice notationally as in (\ref{nsplit}), the integrated B\"acklund chain will be expressed as $u_{\{1,\ldots,N\}}=g/f$ where
\begin{equation}
f = \prod_{i\in\{1,\ldots,N\}}\left.\left[a_i\oT_i+b_i\oT_i^{-1}\right]\tau\right|_{\mathbb{Z}^d}, \quad 
g = \prod_{i\in\{1,\ldots,N\}}\left.\left[a_i\oT_i+b_i\oT_i^{-1}\right]u\tau\right|_{\mathbb{Z}^d}.
\end{equation}

\subsection{$F_I$ solution}
The consistent system on $\mathbb{Z}^d$, where indeterminates are on lattice edges,
\begin{equation}\label{F1}
x\wh{x}q^2=y\wt{y}p^2, \quad (1-x)(1-\wh{x})(1-q^2)=(1-y)(1-\wt{y})(1-p^2),
\end{equation}
is the quadrirational map $F_I$ \cite{ABSf}.
For any solution $u$ of equation $A2$ (\ref{A2}), the association
\begin{equation}\label{subs}
x=pu\wt{u}, \quad y=qu\wh{u},
\end{equation}
defines a solution of (\ref{F1}) \cite{PTV}.
This is verified by substitution of (\ref{subs}); the first equation in (\ref{F1}) is satisfied identically, and the second recovers (\ref{A2}).

The multidimensional solution of (\ref{F1}) obtained from the solution given for (\ref{A2}) therefore takes the form
\begin{equation}
\begin{split}
x = k\frac{\cn(\alpha)g\wt{g}}{\dn(\alpha)f\wt{f}},\\
y = k\frac{\cn(\beta)g\wh{g}}{\dn(\beta)f\wh{f}},
\end{split}
\end{equation}
where the parameters $p,q$ appearing in the equation are related to $\alpha$ and $\beta$ by (\ref{eparams}).
\subsection{$Q4^*$ solution}
Consider now systems (\ref{A2}) and (\ref{F1}) in at least $\mathbb{Z}^d$ with $d\ge 3$, and complement (\ref{subs}) with an equation in the third direction,
\begin{equation}
z=ru\wb{u}.
\end{equation}
It was shown in \cite{ib} that $z$ satisfies a multi-quadratic quad-equation of the class defined in \cite{AtkNie}, say,
\begin{equation}\label{gf}
S_{p^2,q^2}(z,\wt{z},\wh{z},\th{z})=0,
\end{equation}
which can be viewed as superposition principle for B\"acklund transformations of (\ref{F1}).
The multi-quadratic polynomial $S$ is labelled by values $(0,1,\infty,r^2)$, which are branch points of an underlying elliptic curve
\begin{equation}
\{(x,X):X^2=(x-r^2)x(x-1)\},\quad r^2 {\textrm{ distinguished}},
\end{equation}
from which $S$ is uniquely determined, provided the distinguished branch-point is simple.
Transformation to the Jacobi form of the curve,
\begin{equation}\label{jc}
\{(x,X):X^2=(x-c)(x+c)(x-1/c)(x+1/c)\},\quad c {\textrm{ distinguished}},
\end{equation}
maps $S$ to the canonical form $S'$ for Q4$^*$.
The M\"obius transformation $\mm$ is characterised by its action $\mm:(0,1,\infty)\rightarrow (-c,1/c,-1/c)$, and takes the form
\begin{equation}
\mm(z)=\frac{2c^2-(1+c^2)z}{c(1+c^2)z-2c}.\label{mm}
\end{equation}
The modulus associated with this Jacobi curve (\ref{jc}), $c^2$, is related to the lattice parameter $r$ by the equation
\begin{equation}\label{cg}
\mm(r^2)=c \quad \Leftrightarrow \quad r^2=\left(\frac{2c}{1+c^2}\right)^2.
\end{equation}
This corresponds to Gauss' transformation \cite{AKH}.
The associations
\begin{equation}
w=\mm(z),\ a=\mm(p^2), \ b=\mm(q^2), \ c=\mm(r^2),
\label{sbs}
\end{equation}
thus connect (\ref{gf}) with the canonical form,
\begin{multline}
S'_{a,b}(w,\wt{w},\wh{w},\th{w})=\\
(a-b)[(c^{-2}a-c^2b)(w\wt{w}-\wh{w}\th{w})^2-(c^{-2}b-c^2a)(w\wh{w}-\wt{w}\th{w})^2]\\
-(a-b)^2[(w+\th{w})^2(1+\wt{w}^2\wh{w}^2)+(\wt{w}+\wh{w})^2(1+w^2\th{w}^2)]\\
+[(w-\th{w})(\wt{w}-\wh{w})(c^{-1}-cab)-2(a-b)(1+w\wt{w}\wh{w}\th{w})]\\
\times[(w-\th{w})(\wt{w}-\wh{w})(c^{-1}ab-c)-2(a-b)(w\th{w}+\wt{w}\wh{w})]=0.
\label{QQ4}
\end{multline}
The solution of equation (\ref{QQ4}) obtained by implementing transformations (\ref{sbs}) on the previously given solution of (\ref{F1}) is:
\begin{equation}\label{QQ4solitons}
w=\frac
{c' f\wb{f}-g\wb{g}}
{c'g\wb{g}-f\wb{f}}, \quad c' = c\varepsilon_1 /k,
\end{equation}
where the parameters $a$, $b$ and $c$ appearing in the equation are given in terms of $\alpha$, $\beta$ and $\gamma$, by
\begin{equation}
\begin{split}
a&=
\frac{\varepsilon_1\dn(\gamma)-\varepsilon_2}{\cn(\gamma)}\times
\frac
{\varepsilon_1\dn^2(\alpha)+\varepsilon_2\dn(\gamma)}
{\varepsilon_1\dn^2(\alpha)-\varepsilon_2\dn(\gamma)},\\
b&=
\frac{\varepsilon_1\dn(\gamma)-\varepsilon_2}{\cn(\gamma)}\times
\frac
{\varepsilon_1\dn^2(\beta)+\varepsilon_2\dn(\gamma)}
{\varepsilon_1\dn^2(\beta)-\varepsilon_2\dn(\gamma)},\\
c&=\frac{\varepsilon_1 \dn(\gamma)+\varepsilon_2}{\cn(\gamma)},
\end{split}
\end{equation}
and where in all of these expressions, $\varepsilon_1$ and $\varepsilon_2$ are any solution of the equations,
\[ \varepsilon_1^2= \frac{1}{k^2}, \quad \varepsilon_2^2 = \frac{{1-k^2}}{k^2}. \]
They correspond to the four possible values for $c$ determined by (\ref{cg}) in the particular case $r=k\cn(\gamma)/\dn(\gamma)$ corresponding to the original relations (\ref{eparams}).

\bibliographystyle{unsrt}
\bibliography{references}

\begin{thebibliography}{10}

\bibitem{AN}
J.~Atkinson and F.~W. Nijhoff.
\newblock A constructive approach to the soliton solutions of integrable
  quadrilateral lattice equations.
\newblock {\em Comm. Math. Phys.}, 299(2):283--304, 2010.

\bibitem{AJ}
J.~Atkinson and N.~Joshi.
\newblock {S}ingular boundary reductions of type-{Q} {A}{B}{S} equations.
\newblock {\em Int. Math. Res. Not.}, 2012:10.1093/imrn/rns024, 2012.

\bibitem{ABS}
V.~E. Adler, A.~I. Bobenko, and Yu.~B. Suris.
\newblock Classification of integrable equations on quad-graphs. {T}he
  consistency approach.
\newblock {\em Comm. Math. Phys.}, 233(3):513--543, 2003.

\bibitem{ABS2}
V.~E. Adler, A.~I. Bobenko, and Yu.~B. Suris.
\newblock Discrete nonlinear hyperbolic equations. {C}lassification of
  integrable cases.
\newblock {\em Funct. Anal. Appl.}, 43(1):3--17, 2009.

\bibitem{hirota-0}
R.~Hirota.
\newblock Nonlinear partial difference equation {I}: {A} difference analogue of
  the {K}d{V} equation.
\newblock {\em J. Phys. Soc. Jp.}, 43:1429--1433, 1977.

\bibitem{kajota}
K.~Kajiwara and Y.~Ohta.
\newblock Bilinearization and {C}asorati determinant solution to the
  non-autonomous discrete {K}d{V} equation.
\newblock {\em Journal of the Physical Society of Japan}, 77(5):054004, 2008.

\bibitem{James2}
J.~Atkinson and N.~Joshi.
\newblock The {S}chwarzian variable associated with discrete {K}d{V}-type
  equations.
\newblock {\em Nonlinearity}, 25(6):1851--1866, 2010.

\bibitem{NC}
F.~W. Nijhoff and H.~W. Capel.
\newblock The discrete {K}orteweg-de {V}ries equation.
\newblock {\em Acta Appl. Math.}, 39(1-3):133--158, 1995.

\bibitem{Weiss0}
J.~Weiss.
\newblock B{\"a}cklund transformations and the painlev{\'e} property.
\newblock In P.~J. Olver and D.~H. Sattinger, editors, {\em Solitons in
  Physics, Mathematics, and Nonlinear Optics}, pages 175--202, New York, NY,
  1990. Springer New York.

\bibitem{DN}
Y.~Dorfman and F.~W. Nijhoff.
\newblock On a $(2+1)$-dimensional version of the {K}richever-{N}ovikov
  equation.
\newblock {\em Phys. Lett. A}, 157:107--112, 1991.

\bibitem{bog}
L.~V. Bogdanov and B.~G. Konopelchenko.
\newblock Analytic-bilinear approach to integrable hierarchies {I}{I}.
  {M}ulticomponent {K}{P} and 2{D} {T}oda hierarchies.
\newblock {\em J. Math. Phys.}, 39:4701--4728, 1998.

\bibitem{NAH}
F.~W. Nijhoff, J.~Atkinson, and J.~Hietarinta.
\newblock Soliton solutions for {A}{B}{S} lattice equations: {I}. {C}auchy
  matrix approach.
\newblock {\em J. Phys. A: Math. Th.}, 42(40):Art. no.404005, 2009.

\bibitem{Adl}
V.~E. Adler.
\newblock B\"acklund transformation for the {K}richever-{N}ovikov equation.
\newblock {\em Int. Math. Res. Not.}, 1:1--4, 1998.

\bibitem{AdVeQ}
V.~E. Adler and A.~P. Veselov.
\newblock Cauchy problem for integrable discrete equations on quad-graphs.
\newblock {\em Acta Appl. Math.}, 84(2):237--262, 2004.

\bibitem{pip}
P.~H. Van-der {K}amp.
\newblock Initial value problems for quad equations.
\newblock {\em Journal of Physics A: Mathematical and theoretical}, 46:095204,
  2013.

\bibitem{Miwa}
T.~Miwa.
\newblock On {H}irota’s difference equations.
\newblock {\em Proc. Japan Acad. Ser. A Math. Sci.}, 58:9--12, 1982.

\bibitem{HZ}
J.~Hietarinta and D.~J. Zhang.
\newblock Soliton solutions for {A}{B}{S} lattice equations: {I}{I}.
  {C}asoratians and bilinearization.
\newblock {\em J. Phys. A: Math. Th.}, 42(40):Art. no. 404006, 2009.

\bibitem{WZM}
J.~Wang, D-J. Zhang, and K-i. Maruno.
\newblock Connection between the symmetric discrete {A}{K}{P} system and
  bilinear {A}{B}{S} lattice equations.
\newblock {\em Physica D: Nonlinear Phenomena}, 462:134155, 2024.

\bibitem{weiss}
J.~Weiss.
\newblock Periodic fixed points of {B}\"acklund transformations.
\newblock {\em J. Math. Phys.}, 28:2025, 1987.

\bibitem{Hie}
J.~Hieterinta.
\newblock {S}earching for {C}{A}{C}-maps.
\newblock {\em J. Nonl. Math. Phys.}, 12(Suppl. 2):223--30, 2005.

\bibitem{ABSf}
V.~E. Adler, A.~I. Bobenko, and Yu.~B. Suris.
\newblock Geometry of {Y}ang-{B}axter maps: pencils of conics and
  quadrirational mappings.
\newblock {\em Comm. Anal. Geom.}, 12(5):967--1007, 2004.

\bibitem{PTV}
V.~G. Papageorgiou, A.~G. Tongas, and A.~P. Veselov.
\newblock Yang-{B}axter maps and symmetries of integrable equations on
  quad-graphs.
\newblock {\em J. Math. Phys.}, 47:Art. no. 083502, 2006.

\bibitem{ib}
J.~Atkinson.
\newblock Idempotent biquadratics, {Y}ang-{B}axter maps and birational
  representations of {C}oxeter groups.
\newblock {\em arxiv:1301.4613 [nlin.SI]}, 2013.

\bibitem{AtkNie}
J.~Atkinson and M.~Nieszporski.
\newblock Multi-quadratic quad equations: integrable cases from a factorised
  discriminant hypothesis.
\newblock {\em Intl. Math. Res. Not}, 2012.
\newblock doi: 10.1093/imrn/rnt066.

\bibitem{AKH}
N.~I. Akhiezer.
\newblock {\em Elements of the theory of elliptic functions}, volume~79 of {\em
  AMS Translations of mathematical monographs}.
\newblock 1970.
\newblock (translated from the Russian by McFaden H H and edited by Silver B
  1990).

\end{thebibliography}
\end{document}